\documentclass[envcountsect,runningheads,a4paper]{llncs}
\usepackage{comment}
\usepackage{tikz}
\usetikzlibrary{calc}
\usetikzlibrary{positioning}
\usetikzlibrary{snakes}

\usepackage{amssymb,stmaryrd,amsmath}
\usepackage{mdwlist} 
\usepackage{float}   

\spnewtheorem{thrm}{Theorem}[section]{\sc }{\itshape}
\spnewtheorem{lemm}[thrm]{Lemma}{\sc }{\itshape}
\spnewtheorem{prop}[thrm]{Proposition}{\sc }{\itshape}
\spnewtheorem{corr}[thrm]{Corollary}{\sc }{\itshape}

\spnewtheorem{nttn}[thrm]{Notation}{\sc }{}
\spnewtheorem{defn}[thrm]{Definition}{\sc }{}
\spnewtheorem{xmpl}[thrm]{Example}{\sc }{}
\spnewtheorem{rmrk}[thrm]{Remark}{\sc }{}

\usepackage{color}
\definecolor{mygreen}{rgb}{0,0.6,0}
\definecolor{mygray}{rgb}{0.5,0.5,0.5}
\definecolor{mymauve}{rgb}{0.58,0,0.82}

\usepackage{listings}

\definecolor{gray}{RGB}{128, 128, 128}
\definecolor{lightgray}{RGB}{200, 200, 200}
\definecolor{cyan}{RGB}{0, 255, 255}
\definecolor{blue}{RGB}{0, 0, 255}
\definecolor{red}{RGB}{255, 0, 0}
\definecolor{pink}{RGB}{255, 128, 128}
\definecolor{green}{RGB}{0, 128, 0}
\definecolor{lightyellow}{RGB}{255, 255, 200}
\definecolor{purple}{RGB}{128, 0, 128}

\lstdefinestyle{all}
    {basicstyle=\ttfamily\scriptsize,
     keywordstyle=\color{blue}\ttfamily\scriptsize,
     commentstyle=\color{green}\ttfamily\scriptsize,
     stringstyle=\color{red}\ttfamily\scriptsize}

\lstdefinelanguage{hask}{%
    frame=none,
    xleftmargin=2pt,
    belowcaptionskip=\bigskipamount,
    captionpos=b,
    tabsize=2,
    numbers=left,
    numberstyle=\tiny\color{gray},
    emphstyle={\bf},
	morecomment=[s][\color{green}]{\{-}{-\}},
    stringstyle=\mdseries\rmfamily,
    commentstyle=\color{green},
    keywords={},
    keywords=[1]{case, of, data, if, then, else, where, let, in, do},
    keywords=[2]{Chip, Config, CurrencySymbol, TokenName, PubKeyHash, Integer, Value, State, Action, Text, Maybe, Void, TxConstraints,  Contract},
    keywords=[3]{HasNative},
    keywords=[4]{=>},
    keywords=[5]{Just, Nothing, MkChip, MkConfig, SetPrice, Buy},
    keywordstyle=[1]\mdseries\sffamily\color{red},
    keywordstyle=[2]\mdseries\sffamily\color{blue},
    keywordstyle=[3]\mdseries\sffamily\color{green},
    keywordstyle=[4]\mdseries\sffamily,
    keywordstyle=[5]\mdseries\sffamily\color{purple},
    columns=flexible,
    basicstyle=\small\sffamily,
    showstringspaces=false,
    breaklines=false,
    showspaces=false,
    escapeinside={--}{\^^M},escapebegin={\color{green}--},escapeend={},
    literate= {+}{{$+$}}1 {/}{{$/$}}1 {*}{{$*$}}1 {=}{{$=$}}1
              {>}{{$>$}}1 {<}{{$<$}}1 {\\}{{$\lambda$}}1
              {\\\\}{{\char`\\\char`\\}}1
              {->}{{$\rightarrow$}}2 {>=}{{$\geq$}}2 {<-}{{$\leftarrow$}}2
              {<=}{{$\leq$}}2 {=>}{{$\Rightarrow$}}2
              {\ .}{{$\circ$}}2 {\ .\ }{{$\circ$}}2
              {>>}{{>>}}2 {>>=}{{>>=}}2
              {|}{{$\mid$}}1
              {\_}{{\underline{\hspace{2mm}}}}2
}

\lstdefinelanguage{solidity}{%
    frame=none,
    xleftmargin=2pt,
    belowcaptionskip=\bigskipamount,
    captionpos=b,
    tabsize=2,
    numbers=left,
    numberstyle=\tiny\color{gray},
    emphstyle={\bf},
	morecomment=[s][\color{green}]{\{-}{-\}},
    stringstyle=\mdseries\rmfamily,
    commentstyle=\color{green},
    keywords={},
    keywords=[1]{pragma, solidity, contract, event, constructor, require, function, return, emit},
    keywords=[2]{address, uint, mapping},
    keywords=[3]{public, payable, external, view, returns},
    keywords=[4]{=>, +=, -=, =, <=, ==},
    keywords=[5]{msg, sender, transfer, value},
    keywordstyle=[1]\mdseries\sffamily\color{red},
    keywordstyle=[2]\mdseries\sffamily\color{blue},
    keywordstyle=[3]\mdseries\sffamily\color{green},
    keywordstyle=[4]\mdseries\sffamily,
    keywordstyle=[5]\mdseries\sffamily\color{purple},
    columns=flexible,
    basicstyle=\small\sffamily,
    showstringspaces=false,
    breaklines=false,
    showspaces=false,
    escapeinside={--}{\^^M},escapebegin={\color{green}--},escapeend={},
    literate= {+}{{$+$}}1 {/}{{$/$}}1 {*}{{$*$}}1 {=}{{$=$}}1
              {>}{{$>$}}1 {<}{{$<$}}1 {\\}{{$\lambda$}}1
              {\\\\}{{\char`\\\char`\\}}1
              {->}{{$\rightarrow$}}2 {>=}{{$\geq$}}2 {<-}{{$\leftarrow$}}2
              {<=}{{$\leq$}}2 {=>}{{$\Rightarrow$}}2
              {\ .}{{$\circ$}}2 {\ .\ }{{$\circ$}}2
              {>>}{{>>}}2 {>>=}{{>>=}}2
              {|}{{$\mid$}}1
              {\_}{{\underline{\hspace{2mm}}}}2
}

\newcommand\deffont[1]{{\bfseries #1}}
\newcommand\powerset{\f{pow}}

\newcommand\finpow{\f{fin}}
\newcommand\finto{\stackrel{\f{fin}}{\rightharpoonup}}
\newcommand\f[1]{\mathit{#1}}
\newcommand\tf[1]{\mathsf{#1}}

\newcommand\at{\text{@}}
\newcommand\tx{\f{tx}}
\newcommand\txs{\f{T\hspace{-2.1pt}xs}}
\newcommand\inlinehask[1]{\lstinline[language=hask]{#1}}
\newcommand\inlinesolidity[1]{\lstinline[language=solidity]{#1}}
\newcommand\utxo{\f{UTxO}}
\newcommand\valid{\f{valid}}
\newcommand\liff{\Longleftrightarrow}

\newcommand\aeq{\mathrel{=_{\alpha}}}
\newcommand\ssm{{{:}\text{=}}}
\newcommand\minus{{\text{-}}}

\begin{document}
\title{UTxO- vs account-based smart contract blockchain programming paradigms}
\author{Lars Br\"unjes\inst{1} \and Murdoch J. Gabbay\inst{2}}
\institute{IOHK \and Heriot-Watt University, Scotland, UK}
\date{}
\maketitle
\begin{abstract}
We implement two versions of a simple but illustrative smart contract: one in Solidity on the Ethereum blockchain platform, and one in Plutus on the Cardano platform, with annotated code excerpts and with source code attached.
We get a clearer view of the Cardano programming model in particular by introducing a novel mathematical abstraction which we call Idealised EUTxO.
For each version of the contract, we trace how the architectures of the underlying platforms and their mathematics affects the natural programming styles and natural classes of errors.  We prove some simple but novel results about alpha-conversion and observational equivalence for Cardano, and explain why Ethereum does not have them.
We conclude with a wide-ranging and detailed discussion in the light of the examples, mathematical model, and mathematical results so far.
\end{abstract}
\thispagestyle{empty}

\section{Introduction}

In the context of blockchain and cryptocurrencies, smart contracts are a way to make the blockchain programmable.
That is: a smart contract is a program that runs on the blockchain to extend its capabilities.

For the smart contract, the blockchain is just an abstract machine (database, if we prefer) with which it programmatically interacts.
Basic design choices in the blockchain's design can affect the the smart contract programming paradigm which it naturally supports, and this can have far-reaching consequences: different programming paradigms are susceptible to different programming styles, and different kinds of program errors.

Thus, a decision in the blockchain's design can have lasting, unavoidable, and critical effects on its programmability.
It is worth being very aware of how this can play out, not least because almost by definition, applications of smart-contracts-the-programming-paradigm tend to be safety-critical.

In this paper we will consider a simple but illustrative example of a smart contract: a fungible tradable token issued by an issuer who creates an initial supply and then retains control of its price---imitating a government-issued fiat currency, but run from a blockchain instead of a central bank.

We will put this example in the context of two major smart contract languages: Solidity, which runs on the Ethereum blockchain, whose native token is \emph{ether}; and Plutus, which runs on the Cardano blockchain, whose native token is \emph{ada}.\footnote{Plutus and Cardano are IOHK designs.  The CEO and co-founder of IOHK, Charles Hoskinson, was also one of the co-founders of Ethereum.}
We compare and contrast the blockchains in detail and exhibit their respective smart contracts.
Both contracts run, but their constructions are different, in ways that illuminate the essential natures of the respective underlying blockchains and the programming styles that they support.

We will also see that the Ethereum smart contract is arguably buggy,
in a way that flows from the underlying programming paradigm of Solidity/Ethereum.
So even in a simple example, the essential natures of the underlying systems are felt, and with a mission-critical impact.

\begin{defn}
\label{defn.the.spec}
We will use Solidity and Plutus to code a system as follows:
\begin{enumerate*}
\item
An issuer $\mathtt{Issuer}$ creates some initial supply of a tradable token on the blockchain at some location in $\mathtt{Issuer}$'s control; call this the \emph{official portal}.
\item
Other parties buy the token from the official portal at a per-token ether/ada \emph{official price}, controlled by the issuer.\footnote{Think: central bank, manufacturer's price, official exchange rate, etc.}

Once other parties get some token, they can trade it amongst themselves (e.g. for ether/ada), independently of the official portal and on whatever terms and at whatever price they mutually agree.
\item\label{set.price}
The issuer can update the ether/ada official price of the token on the official portal, at any time.
\item\label{initial.supply}
For simplicity, we permit just one initial issuance of the token,
though tokens can be redistributed as just described.
\end{enumerate*}
\end{defn}

\section{Idealised EUTxO}
\label{sect.idealised.cardano}

\subsection{The structure of an idealised blockchain}

We start with a novel mathematical idealisation of the EUTxO (Extended UTxO) model on which Cardano is based~\cite{chakraverty:extum,coutts:forscw}.

\begin{nttn}
\label{nttn.pointed}
Suppose $X$ and $Y$ are sets.
Then:
\begin{enumerate*}
\item
Write $\mathbb N=\{0,1,2,\dots\}$ and $\mathbb N_{>0}=\{1,2,3,\dots\}$.
\item\label{pointed.finite.set}
Write $\finpow(X)$ for the finite powerset of $X$ and $\finpow_!(X)$ for the pointed finite powerset
(the $(X',x)\in \finpow(X)\times X$ such that $x\in X'$).\footnote{This use of `pointed' is unrelated to the `points to' of Notation~\ref{nttn.points.to}.}
\item
Write $\powerset(X)$ for the powerset of $X$, and
$X\finto Y$ for finite maps from $X$ to $Y$ (finite partial functions).
\end{enumerate*}
\end{nttn}

\begin{defn}
Let the \deffont{types} of \deffont{Idealised EUTxO} be a solution to the 
equations in Figure~\ref{fig.ieutxo.types}.\footnote{We write `a solution of' because Figure~\ref{fig.ieutxo.types} does not specify a unique subset for $\tf{Validator}$.
\emph{Computable} subsets is one candidate, but our mathematical abstraction is agnostic to this choice.  This is just the same as function-types not being modelled by \emph{all} functions, or indeed as models of set theory can have different powersets (e.g. depending on whether a powerset includes the Axiom of Choice).}
For diagrams and examples see Example~\ref{xmpl.example.transactions}.
\end{defn}

\begin{figure}[tp]
$$
\begin{array}{r@{\ }l}
\tf{Redeemer}=&\tf{CurrencySymbol}=\tf{TokenName}=\tf{Position}=\mathbb N
\\
\tf{Chip}=&\tf{CurrencySymbol}\times\tf{TokenName}
\\
\tf{Datum}\times\tf{Value}=&\mathbb N\times(\tf{Chip}\finto\mathbb N_{>0})
\\
\tf{Validator}\subseteq&\powerset(\tf{Redeemer}\times\tf{Datum}\times\tf{Value}\times\tf{Context})
\\
\tf{Input}=&\tf{Position}\times\tf{Redeemer}
\\
\tf{Output}=&\tf{Position}\times\tf{Validator}\times\tf{Datum}\times\tf{Value}
\\
\tf{Transaction}=&\finpow(\tf{Input})\times\finpow(\tf{Output})
\\
\tf{Context}=&\finpow_!(\tf{Input})\times\finpow(\tf{Output})
\end{array}
$$
\caption{Types for Idealised EUTxO}
\label{fig.ieutxo.types}
\ \\
\begin{lstlisting}[language=hask]
data Chip = MkChip  
  { cSymbol  :: !CurrencySymbol 
  , cName    :: !TokenName }

data Config = MkConfig
  { cIssuer                  :: !PubKeyHash
  , cTradedChip, cStateChip  :: !Chip }

tradedChip :: Config -> Integer -> Value
tradedChip MkConfig{..} n  =  singletonValue cTradedChip n

data Action = 
    SetPrice  !Integer              
  | Buy       !Integer

transition  :: Config -> State Integer -> Action 
            -> Maybe (TxConstraints Void Void, State Integer)
transition c s (SetPrice p)                    -- ACTION: set price to p
  | p < 0      =  Nothing                      -- p negative? ignore! 
  | otherwise  =  Just                         -- otherwise 
      ( mustBeSignedBy (cIssuer c)             -- issuer signed?
      , s{stateData = p})                      -- set new price!
transition c s (Buy m)                         -- ACTION: buy m chips
  | m <= 0     =  Nothing                      -- buy negative quantity? ignore!
  | otherwise  =  Just                         -- otherwise 
      ( mustPayToPubKey (cIssuer c) value'     -- issuer been paid?
      , s{stateValue = stateValue s - sold})   -- sell chips!
 where
  value'  =  lovelaceValueOf (m * stateData s)  -- final value buyer pays 
  sold    =  tradedChip c m                     -- no. chips buyer gets 

guarded  :: HasNative s => Config -> Integer -> Integer 
         -> Contract s Text ()
guarded c n maxPrice =
  void $ withError $ runGuardedStep (client c) (Buy n) $ 
  \_ p _ -> if p <= maxPrice then Nothing else Just () 
\end{lstlisting}
\caption{Plutus implementation of the tradable token}
\label{fig.plutus.code}
\end{figure}

\begin{rmrk}
\label{rmrk.think.of}
\begin{enumerate}
\item
Think of $r\in\tf{Redeemer}$ as a key, required as a necessary condition by a validator (below) to permit computation.
\item\label{what.is.a.chip}
A chip $c=(d,n)$ is intuitively a currency unit (\textsterling, \textdollar, \dots), where $d\in\tf{CurrencySymbol}$ is assumed to G\"odel encode\footnote{G\"odel encoding refers to the idea of enumerating a countable datatype (in some arbitrary way) so that each element is represented by a unique numerical index.}
some predicate defining a \deffont{monetary policy} (more on this in Remark~\ref{rmrk.messy.reindexing}(\ref{monetary.policy}))
and $n\in\tf{TokenName}$ is just a symbol (Ada has special status and is encoded as $(0,0)$).
\item
$\tf{Datum}$ is any data; we set it to be $\mathbb N$ for simplicity.
A value $v\in\tf{Value}$ is intuitively a multiset of tokens; $v\,c$ is the number of $c$s in $v$.
\begin{itemize*}
\item
We may abuse notation and define $v\,c=0$ if $c\not\in\f{dom}(v)$.
\item
If $\f{dom}(v)=\{c\}$
then we may call $v$ a \deffont{singleton} value (note $v\,c$ need not equal~1).
See line~10 of Figure~\ref{fig.plutus.code} for the corresponding code.
\end{itemize*}
\item
A transaction is a set of inputs and a set of outputs (either of which may be empty).
In a blockchain these are subject to consistency conditions (Definition~\ref{defn.blockchain}); for now we just
say its inputs `consume' some previous outputs, and `generate' new outputs.
A context is just a transaction viewed from a particular input (see next item).
\item
Mathematically a validator $V\in\tf{Validator}$ is the set of $\tf{Redeemer}\times\tf{Datum}\times\tf{Value}\times\tf{Transaction}$
tuples it validates
\emph{but} in the implementation we intend that $V$ is represented by code $\mathtt{V}$ such that
\begin{itemize*}
\item
from $\mathtt{V}$ we cannot efficiently compute a tuple $t$ such that $\mathtt{V}(t){=}\mathtt{True}$, and
\item
from $\mathtt{V}$ and $t$, we can efficiently check if $\mathtt{V}(t){=}\mathtt{True}$.\footnote{The `crypto' in `cryptocurrency' lives here.}
\end{itemize*}
We use a \emph{pointed} transaction (Notation~\ref{nttn.pointed}(\ref{pointed.finite.set}))
because a $\tf{Validator}$ in an $\tf{Output}$ in a $\tf{Transaction}$ is always invoked by some particular $\tf{Input}$ in a later $\tf{Transaction}$ (see Definition~\ref{defn.blockchain}(\ref{blockchain.in.out}\&\ref{blockchain.the.point})), and that invoking input is identified by using $\tf{Context}$.
If $t\in V$ then we say $V$ \deffont{validates} $t$.
\end{enumerate}
\end{rmrk}

\begin{nttn}
\label{nttn.points.to}
\begin{enumerate*}
\item
If $\tx=(I,O)\in\tf{Transaction}$ and $o\in\tf{Output}$, say $o$ \deffont{appears in} $\tx$ and write $o\in \tx$ when $o\in O$; similarly for an input $i\in\tf{Input}$.
We may silently extend this notation to larger data structures, writing for example $o\in\txs$ (Definition~\ref{defn.cong}(\ref{spent.output})).
\item
If $i$ and $o$ have the same position then say that $i$ \deffont{points to} $o$.
\item\label{points.to}
If $\tx=(I,O)\in\tf{Transaction}$ and $i\in I$ then write $\tx\at i$ for the context $((I,i),O)$ obtained by pointing $I$ at $i\in I$.
\end{enumerate*}
\end{nttn}

\begin{defn}
\label{defn.blockchain}
A \deffont{valid blockchain}, or just \deffont{blockchain}, of idealised EUTxO is a finite sequence of transactions $\txs$ such that:
\begin{enumerate*}
\item
Distinct outputs appearing in $\txs$ have distinct positions.
\item\label{blockchain.in.out}
Every input $i$ in some $\tx$ in $\txs$ points to a unique output in some earlier transaction --- it follows from this and condition~1 of this Definition, that distinct inputs appearing in $\txs$ also have distinct positions.

Write this unique output $\txs(i)$.
\item\label{blockchain.the.point}
If $i=(p,k)$ appears in $\tx$ in $\txs$ and points to an earlier output $\txs(i)=(p,V,s,v)$, then $(k,s,\tx\at i)\in V$ ($\at$ from Notation~\ref{nttn.points.to}(\ref{points.to})).
\end{enumerate*}
\end{defn}

\newcommand\nameOaa{$a$}
\newcommand\nameOab{$b$}
\newcommand\nameOac{$c$}
\newcommand\nameIba{\nameOab}
\newcommand\nameObb{$d$}
\newcommand\nameIca{\nameOaa}
\newcommand\nameOca{$e$}
\newcommand\nameOcb{$f$}
\newcommand\nameOcc{$g$}
\newcommand\nameIda{\nameOca}
\newcommand\nameIdb{\nameOcb}
\newcommand\nameIdc{\nameObb}
\newcommand\nameOda{$h$}
\newcommand\nameOdb{$i$}
\newcommand\nameOdc{$j$}
\newcommand\nameOdd{$k$}

\newcommand\drawCircle[2]{
    \path[fill, #1]    (#2) circle[radius=0.3];
    \path[fill, white] (#2) circle[radius=0.15];
}
\newcommand\drawBox[2]{
        \draw[ultra thick, fill=lightgray]
            ($#1 + (-1,-1)$)
            rectangle ($#1 + (1,1)$);
        \node at #1 {$#2$};
}
\newcommand\drawLine[4]{
        \draw[ultra thick] (#1) to[out=0, in=180] node[#2]{#3} (#4);
}
\newcommand\drawLineToCircle[5]{
        \drawLine{#1}{#2}{#3}{#4}
        \drawCircle{#5}{#4}
}
\newcommand\drawLineFromCircle[5]{
        \drawLine{#1}{#2}{#3}{#4}
        \drawCircle{#5}{#1}
}

\newcommand\mkBackbone{
        \coordinate (a) at (0*\ma,0);    
        \coordinate (d) at (3*\ma,0);    

        \coordinate (ab1) at (\mb, \mc);  
        \coordinate (ab2) at (\mb, 0);  
        \coordinate (ab3) at (\mb,-\mc);  

        \coordinate (bc1) at (1.5*\ma,0.5*\mc); 

        \coordinate (de1) at (3.5*\ma, \mc);
        \coordinate (de2) at (3.5*\ma, 1/3*\mc);
        \coordinate (de3) at (3.5*\ma,-1/3*\mc);
        \coordinate (de4) at (3.5*\ma, -\mc);
        
        \coordinate (x2) at (2.5*\ma,-\mc);

        \drawBox{(a)}{\tx_1}
        \drawBox{(b)}{\tx_2}
        \drawBox{(c)}{\tx_3}
        \drawBox{(d)}{\tx_4}
        
        \drawLine{$(a) + (1, 0.5)$}{above, xshift=-1, yshift=1}{\nameOaa}{ab1}
        \drawLineToCircle{$(d) + (1, 0.6)$}{above}{\nameOda}{de1}{red}
        \drawLineToCircle{$(d) + (1, 0.2)$}{above right}{\nameOdb}{de2}{red}
        \drawLineToCircle{$(d) + (1,-0.2)$}{below right}{\nameOdc}{de3}{red}
        \drawLineToCircle{$(d) + (1,-0.6)$}{below}{\nameOdd}{de4}{red}
        \drawLine{$(c) + (1, 0.5)$}{above, xshift=-1}{\nameOca}{cd1}
        \drawLineToCircle{$(c) + (1,-0.5)$}{below, xshift=-1}{\nameOcc}{cd3}{red}

}
\newcommand\mkBackboneBC[1]{
        \coordinate (bc2) at (1.5*\ma-#1*\ma,-\mc);    
        \coordinate (bc3) at (1.5*\ma, 0);

        \coordinate (cd4) at (2.5*\ma+#1*\ma,-\mc);
        
        \coordinate (cd1) at (2.5*\ma, \mc);  
        \coordinate (cd2) at (2.5*\ma, 0);    
        \coordinate (cd3) at (2.5*\ma-#1*\ma,-2/3*\mc+#1*\mc);    

        \coordinate (b) at (1*\ma,0);    
        \coordinate (c) at (2*\ma,0);    
        \coordinate (x1) at (1.5*\ma+#1*\ma,\mc);     

        \mkBackbone

        \drawLineToCircle{$(a) + (1,-0.5)$}{below, xshift=-1, yshift=-1}{\nameOac}{ab3}{red}

        \drawLine{ab2}{above, yshift=-1}{\nameIba}{$(b) + (-1, 0)$}
        \drawLineToCircle{$(a) + (1, 0  )$}{above, yshift=-1}{\nameOab}{ab2}{blue}
   
        \drawLine{$(c) + (1, 0  )$}{above, yshift=-1.5}{\nameOcb}{cd2};

}
\newcommand\mkBackboneCB[1]{
        \coordinate (bc2) at (1.5*\ma,-\mc);    

        \coordinate (bd1) at (2.5*\ma, \mc);
        \coordinate (bd2) at (2.5*\ma,-\mc);
        \coordinate (bd3) at (2.5*\ma, 3/5*\mc);
        \coordinate (bd4) at (2.5*\ma,-1/3*\mc);
 
        \coordinate (cd1) at (1.5*\ma-#1*\ma, \mc);
        \coordinate (cd2) at (1.5*\ma-#1*\ma, 3/5*\mc);
        \coordinate (cd3) at (1.5*\ma-#1*\ma,-2/5*\mc);
        
        \coordinate (c) at (1*\ma,0);       
        \coordinate (b) at (2*\ma,0);       
        \coordinate (x1) at (1.5*\ma,\mc);  

        \mkBackbone

        \drawLineToCircle{$(a) + (1,-0.5)$}{below, xshift=-1, yshift=-1}{\nameOab}{ab3}{red}
        
        \drawLineToCircle{$(a) + (1, 0  )$}{above, yshift=-1}{\nameOac}{ab2}{red}

        \drawLine{$(c) + (1, 0  )$}{below, yshift=1}{\nameOcb}{cd2}

        \coordinate (cb1) at (1.5*\ma+#1*\ma, -1*\mc);
        \drawLine{cb1}{above, xshift=-2}{\nameIba}{$(b) + (-1, 0)$}
}

\begin{figure}[tph]
    \centering
    \newcommand\ma{10}
    \newcommand\mb{5}
    \newcommand\mc{3}
    \begin{tikzpicture}[scale=0.3]
        \coordinate (tx) at (  0, 0);
        \coordinate (i1) at (-\ma, \mb);
        \coordinate (i2) at (-\ma, 0);
        \coordinate (i3) at (-\ma,-\mb);
        \coordinate (o1) at ( \ma, \mc);
        \coordinate (o2) at ( \ma,-\mc);

        \drawBox{(tx)}{tx}

        \drawLineFromCircle{i1}{above right}{$(a,r_1)$}{$(tx) + (-1, 0.5)$}{red}
        \drawLineFromCircle{i2}{above, xshift=-5}{$(b,r_2)$}{$(tx) + (-1, 0  )$}{red}
        \drawLineFromCircle{i3}{below right}{$(c,r_3)$}{$(tx) + (-1,-0.5)$}{red}
        \drawLineToCircle{$(tx) + (1, 0.3)$}{below right, xshift=2, yshift=5}{$(d,V_1,d_1,v_1)$}{o1}{red}
        \drawLineToCircle{$(tx) + (1,-0.3)$}{above right, xshift=2, yshift=-5}{$(e,V_2,d_2,v_2)$}{o2}{red}

    \end{tikzpicture}
    \caption{A transaction $tx$ with three inputs and two outputs, positions $a,b,c,d,e$, redeemers $r_i$, validators $V_j$,
        data $d_j$ and values $v_j$}
    \label{fig.transaction}
    \ \\[5mm]
    \renewcommand\ma{10}
    \renewcommand\mb{5}
    \renewcommand\mc{3}
    \begin{tikzpicture}[scale=0.3]

        \mkBackboneBC{0}

        \drawLine{bc2}{}{}{x2}
        \drawLineToCircle{$(b) + (1,-0.25)$}{below, xshift=-1, yshift=-1}{\nameObb}{bc2}{blue}

        \drawLineFromCircle{ab1}{}{}{x1}{blue}
        \drawLine{x1}{above, xshift=1, yshift=1}{\nameIca}{$(c) + (-1, 0.25)$}

        \drawLineFromCircle{cd1}{above, xshift=1}{\nameIda}{$(d) + (-1, 0.5)$}{blue}
        \drawLineFromCircle{cd2}{above, xshift=-2, yshift=-1}{\nameIdb}{$(d) + (-1, 0  )$}{blue}
        \drawLine{x2}{below}{\nameIdc}{$(d) + (-1,-0.5)$}

    \end{tikzpicture}
    \caption{A blockchain $\mathcal B=[\tx_1,\tx_2,\tx_3,\tx_4]$}
    \label{fig.blockchain}
    \ \\[5mm]
    \begin{tikzpicture}[scale=0.3]
        \mkBackboneBC{0.05}

        \drawLineToCircle{$(b) + (1,-0.25)$}{below, xshift=-1, yshift=-1}{\nameObb}{bc2}{red}
        \drawLineFromCircle{x1}{above, xshift=1, yshift=1}{\nameIca}{$(c) + (-1, 0)$}{red}
        \drawLineFromCircle{cd1}{above, xshift=1, yshift=0}{\nameIda}{$(d) + (-1, 0.5)$}{blue}
        \drawLineFromCircle{cd2}{above,xshift=-1, yshift=-2}{\nameIdb}{$(d) + (-1, 0  )$}{blue}
        \drawLineFromCircle{cd4}{below}{\nameIdc}{$(d) + (-1,-0.5)$}{red}

        \drawCircle{red}{ab1}

    \end{tikzpicture}
    \caption{$B$ chopped up as a blockchain $[\tx_1,\tx_2]$ and a chunk $[\tx_3,\tx_4]$}
    \label{fig.chunks1}
    \ \\[5mm]
    \begin{tikzpicture}[scale=0.3]
        \mkBackboneCB{0.05}

        \drawLine{bd2}{below, xshift=1}{\nameIdc}{$(d) + (-1,-0.5)$}
        \drawLineToCircle{$(b) + (1,-0.25)$}{below, xshift=-1, yshift=-1}{\nameObb}{bd2}{blue}
        \drawLineFromCircle{ab1}{above, xshift=1, yshift=1}{\nameIca}{$(c) + (-1, 0.5)$}{blue}

        \drawLineFromCircle{bd1}{above, xshift=1}{\nameIda}{$(d) + (-1, 0.5)$}{red}
        \drawLineFromCircle{bd3}{below, yshift=0}{\nameIdb}{$(d) + (-1, 0  )$}{red}

        \drawCircle{red}{cb1}
        \drawCircle{red}{cd1}
        \drawCircle{red}{cd2}
    \end{tikzpicture}
    \caption{$B$ chopped up as a blockchain $[\tx_1,\tx_3]$ and a chunk $[\tx_2,\tx_4]$}
    \label{fig.chunks2}
    \ \\[5mm]
    \begin{tikzpicture}[scale=0.3]
        \mkBackboneCB{0}

        \drawLineFromCircle{ab1}{above, xshift=1, yshift=1}{\nameIca}{$(c) + (-1, 0.5)$}{blue}
        \drawLine{bd1}{above, xshift=1}{\nameIda}{$(d) + (-1, 0.5)$}
        \drawLine{bd3}{below, yshift=0}{\nameIdb}{$(d) + (-1, 0  )$}
        \drawLine{bd2}{below, xshift=1}{\nameIdc}{$(d) + (-1,-0.5)$}
        \drawLineToCircle{$(b) + (1,-0.25)$}{below, xshift=-1, yshift=-1}{\nameObb}{bd2}{blue}

        \drawLineFromCircle{ab3}{}{}{cb1}{blue}
        \drawLineFromCircle{cd1}{}{}{bd1}{blue}
        \drawLineFromCircle{cd2}{}{}{bd3}{blue}

    \end{tikzpicture}
    \caption{The blockchain $\mathcal B'=[\tx_1,\tx_3,\tx_2,\tx_4]$}
    \label{fig.blockchain'}
\end{figure}

\begin{xmpl}
\label{xmpl.example.transactions}
Example transactions, blockchains, and chunks are illustrated in Figures~\ref{fig.transaction}, \ref{fig.blockchain}, \ref{fig.chunks1}, \ref{fig.chunks2}, and~\ref{fig.blockchain'}.

We leave it as an exercise to the reader to verify that: $\mathcal B$, $\mathcal B'$, $[\tx_1,\tx_2]$ and $[\tx_1,\tx_3]$ are blockchains;\ and $[\tx]$, $[\tx_3,\tx_4]$ and $[\tx_2,\tx_4]$ are chunks but not blockchains.
Also, e.g. $[\tx_2,\tx_1]$ is neither blockchain nor chunk, though it is a list of transactions, because the \nameIba-input of $\tx_2$ points to the later \nameOab-output of $\tx_1$.
\end{xmpl}

\begin{nttn}
\label{nttn.valid}
$\txs$ will range over finite sequences of transactions, and $\mathcal B$ will range over blockchains.
We write $\f{valid}(\txs)$ for the assertion ``$\txs$ is a blockchain''.
\end{nttn}

\begin{rmrk}
\label{rmrk.messy.reindexing}
In words:
\emph{A sequence of transactions is a blockchain when every input points to an earlier validating output.}
Note that output and inputs of a transaction may be empty, and this matters because any blockchain must contain a so-called \emph{genesis block}, which is a transaction without inputs.
We call Figure~\ref{fig.ieutxo.types} and Definition~\ref{defn.blockchain} \emph{Idealised EUTxO} because:
\begin{enumerate*}
\item
The implementation has bells and whistles which we omit for simplicity:
\begin{enumerate*}
\item
There is a second type of output for making payments to people.
\item
Values represent `real money' and so are preserved---the sum of input values must be equal to the sum of output values---unless they are not preserved, e.g. due to fees (which reduce the sum of outputs) or forging (creating) tokens (which increase it).
\item\label{slot.ranges}
Transactions can be made time-sensitive using \emph{slot ranges} (Remark~\ref{rmrk.slot.ranges}); a transaction can only be accepted into a block whose slot is inside the transaction's slot range.
\end{enumerate*}
All of the above is important for a working blockchain but for our purposes it would just add complexity.
\item\label{monetary.policy}
We continue the discussion in Remark~\ref{rmrk.think.of}(\ref{what.is.a.chip}).
If we consider a chip $c=(d,n)$, the currency symbol $d$ is not aritrary: it is (a G\"odel encoding of) a \emph{monetary policy} predicate.
In the implementation, only transactions which satisfy all pertinent monetary policies are admissible.

This is relevant to our example in Section~\ref{sect.plutus.implementation} because we will assume a \emph{state chip} with a monetary policy which enforces that it is affine (zero or one chips on the blockchain; see Remark~\ref{rmrk.state.chip}).
Explaining the mechanics of how this works is outside the scope of this paper; here we just note it exists.
\item\label{list.of.outputs}
In the implementation, a transaction is a pair of a set of inputs and a \emph{list} of outputs.
This is because an implementation concerns a run on a particular concrete blockchain $\mathcal B$, and we want to assign concrete positions for outputs in $\mathcal B$; so with a list the output located at the $j$th output of $i$th transaction could get position $2^i3^j$.

However, we care here about the theory of blockchains (plural).
It is better to use sets and leave positions abstract, since if positions are fixed by their location in a blockchain then in Theorem~\ref{thrm.might.as.well} and the lemmas leading up to it, when we rearrange transactions in a blockchain (e.g. proving an observational equivalence) we could have to track an explicit reindexing of positions in the statement of the results.
More on this in Subsection~\ref{subsect.name-carrying}.
\end{enumerate*}
\end{rmrk}

\subsection{UTxOs and observational equivalence}
\label{subsect.observational.equivalence}

We will be most interested in Definition~\ref{defn.cong} when $\txs$ and $\txs'$ are blockchains $\mathcal B$ and $\mathcal B'$.
However, it is convenient for Lemma~\ref{lemm.apart.commute}(\ref{apart.commute}) if we consider the more general case of any finite sequences of transactions $\txs$ and $\txs'$:
\begin{defn}
\label{defn.cong}
\begin{enumerate*}
\item\label{spent.output}
Call an output $o\in\txs$ \deffont{spent} (in $\txs$) when a later input points to it, and otherwise call $o$ \deffont{unspent} (in $\txs$).
\item
Write $\utxo(\txs)$ for the set of unspent outputs in $\txs$.
\item
If $\utxo(\txs)=\utxo(\txs')$ then write $\txs\approx\txs'$ and call $\txs$ and $\txs'$ \deffont{observationally equivalent}.
\end{enumerate*}
\end{defn}

\begin{nttn}
\label{nttn.append}
Given $\txs$ and a $\tx$, write $\txs;\tx$ for the sequence of transactions obtained by appending $\tx$ to $\txs$.
We will mostly care about this when $\txs$ and $\txs;\tx$ are blockchains, and if so this will be clear from context.
\end{nttn}

\begin{lemm}
Validity (Definition~\ref{defn.blockchain})
is closed under initial subsequences, but not necessarily under final subsequences:
\begin{enumerate*}
\item
$\valid(\txs;\tx)$ implies $\valid(\txs)$.
\item
$\valid(\tx;\txs)$ does not necessarily imply $\valid(\txs)$.
\end{enumerate*}
\end{lemm}
\begin{proof}
For part~1: removing $\tx$ cannot invalidate the conditions in Definition~\ref{defn.blockchain}.
For part~2: it may be that an input in $\txs$ points to an output in $\tx$; if we then remove $\tx$ we would violate condition~\ref{blockchain.in.out} of Definition~\ref{defn.blockchain} that every input must point to a previous output.
\end{proof}

A special case of interest is when two transactions operate on non-overlapping parts of the preceding blockchain:
\begin{nttn}
Suppose $\tx$ and $\tx'$ are transactions.
Write $\tx\#\tx'$ when the positions mentioned in the inputs and outputs of $\tx$, are disjoint from those mentioned in $\tx'$, and in this case call $\tx$ and $\tx'$ \deffont{apart}.
Similarly for $\tx\#\txs$.
\end{nttn}

\begin{lemm}
$\tx\#\tx'\liff \tx'\#\tx$.
\end{lemm}
\begin{proof}
An easy structural fact.
\end{proof}

\begin{rmrk}
In Lemma~\ref{lemm.apart.commute} and Theorem~\ref{thrm.might.as.well} below, note
that:
\begin{itemize*}
\item
$\f{valid}(\mathcal B;\tx)$ (Notations~\ref{nttn.valid} and~\ref{nttn.append}) can be read as the assertion ``it is valid to append the transaction $\tx$ to the blockchain $\mathcal B$''.
\item
$\f{valid}(\mathcal B;\txs)$ can be read as ``it is valid to extend $\mathcal B$ with $\txs$''.
\item
If $\tx\#\tx'$ then they mention disjoint sets of positions, so they cannot point to one another and the UTxOs they point to are guaranteed distinct.
\end{itemize*}
\end{rmrk}

\begin{lemm}
\label{lemm.apart.commute}
\begin{enumerate*}
\item\label{apart.commute}
If $\tx\#\tx'$ then we have $\f{valid}(\mathcal B;\tx;\tx')\liff\f{valid}(\mathcal B;\tx';\tx)$ and $\mathcal B;\tx;\tx'\cong\mathcal B;\tx';\tx$.\footnote{We don't necessarily know $\mathcal B;\tx;\tx'$ is a blockchain, which is why we stated Definition~\ref{defn.cong} for sequences of transactions.}
Intuitively: if two transactions are apart then it is valid to commute them.
\item\label{two.valid.apart}
If $\f{valid}(\mathcal B;\tx';\tx)$ then $\f{valid}(\mathcal B;\tx)\liff\tx\#\tx'$.
(Some real work happens in this technical result, and we use this work to prove Theorem~\ref{thrm.might.as.well}.)
\end{enumerate*}
\end{lemm}
\begin{proof}
\begin{enumerate*}
\item
By routine checking Definition~\ref{defn.blockchain}.
\item
Suppose $\neg(\tx\#\tx')$, so some position $p$ is mentioned by $\tx$ and $\tx'$.
If $p$ is in an input in both $\tx$ and $\tx'$, or an output in both, then $\f{valid}(\mathcal B;\tx';\tx)$ is impossible because each input must point to a unique earlier output, and each output must have a unique position.
If $p$ is in an input in $\tx$ and an output in $\tx'$ then $\neg\f{valid}(\mathcal B;\tx)$, because now this input points to a nonexistent output.
If $p$ is in an output in $\tx$ and an input in $\tx'$ then $\f{valid}(\mathcal B;\tx';\tx)$ is impossible, since each input must point to an \emph{earlier} output.

Conversely suppose $\tx\#\tx'$.
Then $\tx'$ must point only to outputs in $\mathcal B$ and removing $\tx$ cannot disconnect them and so cannot invalidate the conditions in Definition~\ref{defn.blockchain}.
\end{enumerate*}
\end{proof}

\begin{rmrk}
\label{rmrk.utxo.stateless}
Theorem~\ref{thrm.might.as.well} below gives a sense in which UTxO-based accounting is `stateless'.
With the machinery we now have the proof looks simple, but this belies its significance, which we now unpack in English:

Suppose we submit $\tx$ to some blockchain $\mathcal B$.
Then \emph{either} the submission of $\tx$ fails and is rejected (e.g. if some validator objects to it)---\emph{or} it succeeds and $\tx$ is appended to $\mathcal B$.

If it succeeds, then even if other transactions $\txs$ get appended first---e.g. they were submitted by other actors whose transactions arrived first, so that in-between us creating $\tx$ and the arrival of $\tx$ at the main blockchain, it grew from $\mathcal B$ to $\mathcal B;\txs$---then the result $\mathcal B;\txs;\tx$ is \emph{up to observational equivalence} equivalent to $\mathcal B;\tx;\txs$, which is what \emph{would} have happened \emph{if} our $\tx$ had arrived at $\mathcal B$ instantly and before the competing transactions $\txs$.

In other words: if we submit $\tx$ to the blockchain, then at worst, other actors' actions might prevent $\tx$ from getting appended, however, \emph{if} our transaction gets onto the blockchain somewhere, then we obtain our originally intended result up to observational equivalence.
\end{rmrk}

\begin{thrm}
\label{thrm.might.as.well}
Suppose $\valid(\mathcal B;\txs;\tx)$ and $\valid(\mathcal B;\tx)$.
Then
\begin{enumerate*}
\item
$\valid(\mathcal B;\tx;\txs)$ and
\item
$\mathcal B;\tx;\txs\cong\mathcal B;\txs;\tx$.
\end{enumerate*}
\end{thrm}
\begin{proof}
Using Lemma~\ref{lemm.apart.commute}(\ref{two.valid.apart}) $\tx\#\txs$.
The rest follows using Lemma~\ref{lemm.apart.commute}(\ref{apart.commute}).
\end{proof}

\begin{rmrk}
\label{rmrk.slot.ranges}
The Cardano implementation has \emph{slot ranges} (Remark~\ref{rmrk.messy.reindexing}(\ref{slot.ranges})).
These introduce a notion of time-sensitivity to transactions which breaks Theorem~\ref{thrm.might.as.well}(1), because $\txs$ might be time-sensitive.
If we enrich Idealised EUTxO with slot ranges then a milder form of the result holds which we sketch as Proposition~\ref{prop.slots} below.
\end{rmrk}

\begin{prop}
\label{prop.slots}
If we extend our blockchains with slot ranges, which restrict validity of transactions to defined time intervals, then Theorem~\ref{thrm.might.as.well} weakens to:
\begin{quote}
If $\valid(\mathcal B;\tx;\txs)$ and $\valid(\mathcal B;\txs;\tx)$ then $\mathcal B;\tx;\txs\cong\mathcal B;\txs;\tx$.
\end{quote}
\end{prop}
\begin{proof}
As for Theorem~\ref{thrm.might.as.well}, noting that slot ranges are orthogonal to UTxOs, provided the transactions concerned can be appended.
\end{proof}

\subsection{$\alpha$-equivalence and more on observational equivalence}
\label{subsect.name-carrying}

Our syntax for $\mathcal B$ in Definition~\ref{defn.blockchain} is \emph{name-carrying}; outputs are identified by unique markers which---while taken from $\mathbb N$; see ``$\tf{Position}{=}\mathbb N$''in Figure~\ref{fig.ieutxo.types}---are clearly used as \emph{atoms} or \emph{names} to identify binding points on the blockchain, to which at most one later input may bind.
Once bound this name can be thought of graphically as an edge from an output to the input that spends it, so clearly the choice of name/position---once it is bound---is irrelevant up to permuting our choices of names.
This is familiar from $\alpha$-equivalence in syntax, where e.g. $\lambda a.\lambda b.ab$ is equivalent to $\lambda b.\lambda a.ba$.
We define:
\begin{defn}
\begin{enumerate*}
\item
Write $\mathcal B\aeq\mathcal B'$ and call $\mathcal B$ and $\mathcal B'$ \deffont{$\alpha$-equivalent} when they differ only in their choice of positions of spent output-input pairs (Definition~\ref{defn.cong}(\ref{spent.output})).\footnote{Positions of unspent outputs (UTxOs) cannot be permuted.  If we permute a UTxO position in $\mathcal B$, we obtain a blockchain $\mathcal B'$ with a symmetric equivalence to $\mathcal B$ but not observationally equivalent to it (much as $\minus i$ relates to $i$ in $\mathbb C$).  More on this in~\cite{gabbay:equzfn}.}
It is a fact that this is an equivalence relation.
\item
If $\Phi$ is an assertion about blockchains, write ``up to $\alpha$-equivalence, $\Phi$'' for the assertion ``there exist $\alpha$-equivalent forms of the arguments of $\Phi$ such that $\Phi$ is true of those arguments''.
\end{enumerate*}
\end{defn}

Lemma~\ref{lemm.utxo.is.local} checks that observational equivalence interacts well with being a valid blockchain and appending transactions.
We sketch its statement and its proof, which is by simple checking.
In words it says: \emph{extensionally, a blockhain up to $\alpha$-equivalence is just its UTxOs}:
\begin{lemm}
\label{lemm.utxo.is.local}
\begin{enumerate*}
\item
$\mathcal B{\cong}\mathcal B'\land\f{valid}(\mathcal B;\tx)\land\f{valid}(\mathcal B';\tx)$ implies $\mathcal B;\tx{\cong}\mathcal B';\tx$.
\item
If $\mathcal B\aeq\mathcal B'$ then $\mathcal B\cong\mathcal B'$.
\item\label{up.to.alpha.1}
Up to $\alpha$-equivalence, if $\mathcal B\cong\mathcal B'$ then $\f{valid}(\mathcal B;\tx)\liff\f{valid}(\mathcal B';\tx)$.
\item\label{up.to.alpha.2}
Up to $\alpha$-equivalence, if $\mathcal B\cong\mathcal B'\land\f{valid}(\mathcal B;\tx)$ then $\mathcal B;\tx\cong\mathcal B';\tx$.
\end{enumerate*}
\end{lemm}

\begin{rmrk}
We need $\alpha$-conversion in cases~\ref{up.to.alpha.1} and~\ref{up.to.alpha.2} of Lemma~\ref{lemm.utxo.is.local} because $\f{valid}(\mathcal B;\tx)$ might fail due to \emph{accidental name-clash} between the position assigned to a spent output in $\mathcal B$, and a position in an unspent output of $\tx$.
We need $\alpha$-conversion to rename the bound position and avoid this name-clash.

This phenomenon is familiar from syntax, e.g. we know to $\alpha$-convert $a$ in $(\lambda a.b)[b\ssm a]$ to obtain (up to $\alpha$-equivalence) $\lambda a'.a$.

There are many approaches to $\alpha$-conversion: graphs, de Bruijn indexes~\cite{bruijn:lamcnn}, name-carrying syntax with an equivalence relation as required (used in this paper), the \emph{nominal abstract syntax} of the second author and others~\cite{gabbay:newaas-jv}, and the type-theoretic approach in~\cite{DBLP:journals/pacmpl/AllaisA0MM18}.
Studying what works best for a structural theory of blockchains is future research.

In this paper we have used raw name-carrying syntax, possibly quotiented by equivalence as above; a more sophisticated development might require more.
Note the implementation solution discussed in Remark~\ref{rmrk.messy.reindexing}(\ref{list.of.outputs}) corresponds to a de Bruijn index approach, and for our needs in this paper, this does \emph{not} solve all problems, as discussed in that Remark (see `messy reindexing').
\end{rmrk}

\section{The Plutus smart contract}
\label{sect.plutus.implementation}

\begin{defn}
Relevant parts of the Plutus code to implement Definition~\ref{defn.the.spec} are in Figure~\ref{fig.plutus.code}.
Full source is at \url{\small https://arxiv.org/src/2003.14271/anc}.
\end{defn}

\begin{rmrk}
\begin{enumerate*}
\item
\inlinehask{Chip} corresponds to $\f{Chip}$ from Figure~\ref{fig.ieutxo.types}.
\item
\inlinehask{Config} stores configuration: the issuer (given by a hash of their public key), the chip traded, and the state chip (see Remark~\ref{rmrk.state.chip}).
\item
\inlinehask{tradedChip} packages up $n$ of \inlinehask{cTradedChip} in a value (think: a roll of quarters).
\item
\inlinehask{Action} is a datatype which stores labels of a transition system, which in our case are either `buy' or `set price'.
\item
\inlinehask{State Integer} corresponds to $\tf{Datum}$ in Figure~\ref{fig.ieutxo.types}.
\inlinehask{stateData s} on line~29 retrieves this datum and uses it as the price of the token.
\item
\inlinehask{value'} (lines~26 and~29) is the amount of ada paid to the Issuer.\footnote{We call it \inlinehask{value} because it directly corresponds with \inlinesolidity{msg.value} in Figure~\ref{fig.solidity.code}, whose name is fixed in Solidity.  We add a dash to avoid name-clash with \inlinehask{value}, an existing function from the Plutus \inlinehask{Ledger} library.}
\end{enumerate*}
\end{rmrk}

\begin{rmrk}[State chip]
\label{rmrk.state.chip}
In Remark~\ref{rmrk.utxo.stateless} we described in what sense Idealised EUTxO (Figure~\ref{fig.ieutxo.types}) is stateless.
Yet Definition~\ref{defn.the.spec}(\ref{set.price}) specifies that Issuer can set a price for the traded chip.
This is seems stateful.
How to reconcile this?

We create a \emph{state chip} \inlinehask{cStateChip}, whose monetary policy (Remark~\ref{rmrk.messy.reindexing}(\ref{monetary.policy}))
enforces that it is \emph{affine} and \emph{monotone increasing}, and thus \emph{linear} once created.
The Issuer issues an initial transaction to the blockchain which sets up our trading system, and creates a unique UTxO that contains this state chip in its value, with the price in its $\tf{Datum}$ field.

The UTxO with the state chip corresponds to the \emph{official portal} from Definition~\ref{defn.the.spec}, and its state datum corresponds to the \emph{official price}.

Monetary policy ensures this is now an invariant of the blockchain as it develops, and anybody can check the current price by looking up the unique UTxO distinguished by carrying precisely on \inlinehask{cStateChip}, and looking at the value in its $\tf{Datum}$ field.\footnote{This technique was developed by the IOHK Plutus team.}
The interested reader can consult the source code.
\end{rmrk}

\begin{rmrk}[Expressivity of off-chain code]
The Plutus contract is a \emph{state machine}, whose transition function \inlinehask{transition} is compiled into a $\tf{Validator}$ function, and so is explicitly \emph{on-chain}.
The function \inlinehask{guarded} lives \emph{off-chain}; e.g. in the user's wallet.
It can construct and send a \inlinehask{Buy}-transaction to (a UTxO with the relevant validator on) the blockchain, after checking that the price is acceptable.

If accepted, the effect of this transaction is independant of concurrent actions of other users, in senses we have made mathematically precise (cf. Remarks~\ref{rmrk.utxo.stateless} and~\ref{rmrk.can.fix}).
This gives Plutus off-chain code a power that Solidity off-chain code cannot attain.
\end{rmrk}

\section{The Solidity smart contract}
\label{sect.solidity.implementation}

\subsection{Description}

\begin{figure}[tph]
$$
\begin{array}{r@{\ }l}
\tf{Sender}=&\tf{FunctionName}=\tf{ContractName}=\tf{Datum}=\tf{Value}=\mathbb N
\\
\tf{Transaction}=&\tf{ContractName}\times\tf{FunctionName}\times\tf{Sender}\times\tf{Value}\times\tf{Datum}
\\
\tf{Contract}=&\tf{ContractName}\times\finpow(\tf{Function})
\\
\tf{Function}\subseteq&\tf{FunctionName}\times \bigl((\tf{Sender}\times\tf{Value}\times\tf{Datum})\to\tf{Blockchain}\to\tf{Blockchain}\bigr)
\\
\tf{Blockchain}=&\tf{Contract}\finto (\tf{Value}\times\tf{Datum})
\end{array}
$$
\caption{Types for Idealised Ethereum}
\label{fig.ethereum.types}
\ \\
\begin{lstlisting}[language=solidity, firstline=3]
contract Changing {
    address payable public issuer;              // issues the token
    uint public price;                          // current price
    mapping (address => uint) public balances;  // tracks who owns how many tokens

    constructor(uint _count, uint _price) public {
        require(_count > 0, "count must be positive");
        require(_price > 0, "price must be positive");
        issuer               =  msg.sender;
        price                =  _price;
        balances[msg.sender] =  _count;
    }

    function send(address _receiver, uint _amount) public { 
        require(_amount <= balances[msg.sender], "balance too low");
        balances[msg.sender] -=  _amount;
        balances[_receiver] +=  _amount;
    }
 
    function buy() public payable {
        uint _tokens = msg.value / price; 
        require(_tokens <= balances[issuer], "not enough tokens");
        issuer.transfer(msg.value);
        balances[issuer]     -=  _tokens;        
        balances[msg.sender] +=  _tokens;
    }
    
    function setPrice(uint _newPrice) public {
        require(msg.sender ==  issuer, "only issuer can set price");
        price = _newPrice;
    }
}
\end{lstlisting}
\caption{Solidity implementation of the tradable coin}
\label{fig.solidity.code}
\end{figure}

\begin{rmrk}
The Ethereum blockchain is account-based: it can be thought of as a state machine whose transitions modify a global state of contracts containing functions and data.
We propose in Figure~\ref{fig.ethereum.types} a type presentation of Idealised Ethereum, parallel to the Idealised EUTxO of Figure~\ref{fig.ieutxo.types}.
In brief:
\begin{enumerate*}
\item
$\tf{Sender}$, $\tf{Address}$, and $\tf{Value}$ are natural numbers $\mathbb N$.
$\tf{Datum}$ is intended to be structured data which we G\"odel encode for convenience.
$\tf{FunctionName}$ and $\tf{ContractName}$ are names, which again for simplicity we encode as $\mathbb N$.
\item
A contract has a name and a finite set of functions (Notation~\ref{nttn.pointed}(\ref{pointed.finite.set})).
\item
A function has a name and maps an input to a blockchain transformer.
\item
A blockchain is a finite collection of contracts, to each of which is assigned a state of a value (a balance of ether) and some structured data.
We intend that a valid blockchain satisfy some consistency conditions, including that each contract in it have a distinct name.

\end{enumerate*}
Thus the contract \inlinesolidity{Changing} (line~1 of Figure~\ref{fig.solidity.code}), once deployed, is located on the Ethereum blockchain, with its functions and state.
This contract and its state are `in one place' on the blockchain; not spread out across multiple UTxOs as in Cardano.
There is no need for the state-chip mechanism from Remark~\ref{rmrk.state.chip}.
\end{rmrk}

\begin{rmrk}
We now briefly read through the code:
\begin{enumerate}
\item
\inlinesolidity{address} is an address for the Issuer.
It is \inlinesolidity{payable} (it can receive ether) and \inlinesolidity{public} (its value can be read by any other function on the blockchain).\footnote{Any data on the Ethereum blockchain is public in the external sense that it can be read off the binary data of the blockchain as a file on a machine running it.
However, not all data is \inlinesolidity{public} in the internal sense that it can be accessed from any code running on the Ethereum virtual machine.}
\item
\inlinesolidity{constructor} is the function initialising the contract. It gives \inlinesolidity{issuer} (who triggers the contract) all the new token.
\item
\inlinesolidity{send} is a function to send money to another address.
Note that this is not in the Plutus code; this is because we got it `for free' as part of Cardano's in-built support for currencies (this is also why Idealised EUTxO has the type $\tf{Value}$ in Figure~\ref{fig.ieutxo.types}).

\item
\inlinesolidity{buy} on line~20 is analogue to the \inlinesolidity{Buy} transition in Figure~\ref{fig.plutus.code}.
\end{enumerate}
\end{rmrk}

\subsection{Discussion}
\label{subsect.discussion}

\begin{rmrk}
\label{rmrk.cannot.fix}
In line~21 of Figure~\ref{fig.solidity.code} we calculate the tokens purchased using a division of \inlinesolidity{value} (the total sum paid) by \inlinesolidity{price}.

In contrast, in line~29 of Figure~\ref{fig.plutus.code} we calculate the sum by multiplying the number of tokens by the price of each token.

There is a reason for this: we cannot perform the multiplication in the Solidity \inlinesolidity{buy} code because we do not actually know the price per token at the time the function acts to transition the blockchain.
We can access a price at the time of invoking \inlinesolidity{buy} by querying \inlinesolidity{Changing.price()}, but by the time that invocation reaches the Ethereum blockchain and acts on it, the blockchain might have undergone transitions, and the price might have changed.

Because Ethereum is stateful and has nothing like Remark~\ref{rmrk.utxo.stateless} and Theorem~\ref{thrm.might.as.well}, this cannot be fixed.
\end{rmrk}

\begin{rmrk}
\label{rmrk.can.fix}
One might counter that this could indeed be fixed, just by changing \inlinesolidity{buy} to include an extra parameter which is the price that the buyer expects to pay, and if this expectation is not met then the function terminates.
However, the issuer controls the contract (\inlinesolidity{issuer{=}\ msg.sender} on line~9 of Figure~\ref{fig.solidity.code}), including the code for \inlinesolidity{buy}, so this safeguard can only exist at the \emph{issuer's} discretion---and our issuer, whether through thoughtlessness or malice, has not included it.

In Cardano the buyer has more control, because by Theorem~\ref{thrm.might.as.well} a party issuing a transaction knows (up to observational equivalence) precisely what the inputs and outputs will be; the only question is whether it successfully attaches to the blockchain (cf. Remark~\ref{rmrk.utxo.stateless} and Subsection~\ref{subsect.name-carrying}).

Another subtle error in the Ethereum code is that the \inlinesolidity{/} on line~21 is integer division, so there may be a rounding error.\footnote{It would be unheard of for such elementary mistakes to slip into production code; and even if it did happen, it is hardly conceivable that such errors would happen repeatedly across a wide variety of programming languages.

That was sarcasm, but the point may bear repeating: programmer error and programming language design are two sides of a single coin.}
\end{rmrk}

\begin{rmrk}
\label{rmrk.ethereum.errors}
The Ethereum code contains two errors, but the emphasis of this discussion is not that it is possible to write buggy code (which is always true) rather, we draw attention to how and why the underlying accounts-based structure of Ethereum invites and provides cover for certain errors \emph{in particular}.

On the other hand, the Plutus code is more complex than the Ethereum code.
Plutus is arguably conceptually beautiful, but the pragmatics as currently implemented are fiddly and the amount of boilerplate required to get our Plutus contract running much exceeds that required for our Ethereum contract.
This may improve as Plutus matures and libraries improve---Plutus was streamlined as this paper was written, in at least one instance because of a suggestion from this paper\footnote{---the need for \inlinehask{runGuardedStep}.}---but it remains to be seen if the gap can be totally closed.

So Ethereum is simple, direct and alluring, but can be dangerous, whereas Plutus places higher burdens on the programmer but enjoys some good mathematical properties which can enhance programmability.
This is the tradeoff.
\end{rmrk}

\begin{rmrk}[Contracts on the blockchain]
\label{rmrk.contracts}
It is convenient to call the code referred to in Figures~\ref{fig.plutus.code} and~\ref{fig.solidity.code} \emph{contracts}, but this is in fact an imprecise use of the term: a contract is an entity on the blockchain, whereas the figures describe \emph{schemas} or \emph{programs} which, when invoked with parameters, attempt to push a contract onto the blockchain.
\begin{itemize*}
\item
For Plutus, each time we invoke the contract schema for some choice of values for \inlinehask{Config}, the schema tries to creates a transaction on the blockchain which generates a UTxO which carries a state-chip \emph{and} an instance of the contract.
The contract is encoded in the monetary policy of the state-chip (which enforces linearity once created) and in the Validator of that UTxO.
The issuer's address is encoded in the Validator field of the UTxO.
A mathematical presentation of the data structures concerned is in Figure~\ref{fig.ieutxo.types}.
\item
For Ethereum, Figure~\ref{fig.solidity.code} also describes a schema which deploys a contract to the blockchain.
Each time we invoke its constructor we generate a specific instance of \inlinesolidity{Changing}.
Configuration data, given by the constructor arguments,
is simply stored as values of the global state variables of that instance, such as \inlinesolidity{issuer} (line~2 of Figure~\ref{fig.solidity.code}).
\end{itemize*}
\end{rmrk}

\begin{rmrk}[State \& possible refinements]
\label{rmrk.plutus.error}
In Solidity, state is located in variables in contracts.
To query a state we just query a variable in a contract of interest, e.g. \inlinesolidity{price} or \inlinesolidity{balances} in contract \inlinesolidity{Changing} in Figure~\ref{fig.solidity.code}.

In Plutus, state can be more expensive.
State is located in unspent outputs; see the $\tf{Datum}$ and $\tf{Value}$ fields in $\tf{Output}$ in Figure~\ref{fig.ieutxo.types}.
In a transaction that queries or modifies state, we consume any UTxOs containing state to which we refer, and produce new UTxOs with new state.
Even if our queries are read-only, they destroy the UTxOs read and create new UTxOs.\footnote{This is distinct from a user inspecting the contents of UTxOs from outside the blockchain, i.e. by reading state off the hard drive of their node or Cardano wallet.}

Suppose a UTxO contains some popular state information, so multiple users submit transactions referencing it.
Only one transaction can succeed and be appended and
the other transactions will \emph{not} link to the new state UTxO which is generated---even if it contains the same state data as the older version.
Instead, the other transactions will fail, because the particular incarnation of the state UTxO to which they happened to link, has been consumed.

This could become a bottleneck.
In the specific case of our Plutus contract, it could become a victim of its own success if our token is purchased
with high enough frequency, since access to the state UTxO and the tokens on it could become contested and slow and expensive in user's time.
The state UTxO becomes, in effect, constantly locked by competing users---though note the the blockchain would still be behaving correctly in the sense that nobody can lose tokens.

For a more scaleable scenario, some library for parallel redundancy with a monetary policy for the state chip(s) implementing a consensus or merging algorithm, might be required.
This is future work.

One more small refinement: every token ever bought must go through the contract at the official price---even if it is the issuer trying to withdraw it (recall from Definition~\ref{defn.the.spec}(\ref{initial.supply}) that we allow only one issuance).
In practice, the issuer might want to code a way to get their tokens out without going through \inlinehask{buy}.
\end{rmrk}

\begin{rmrk}[Resisting DOS attacks]
An aside contrasting how Cardano and Ethereum manage fees to prevent denial-of-service (DOS) attacks:
In Cardano, users add transactions to the blockchain for a fee of $a+bx$ where $a$ and $b$ are constants and $x$ is the size of the transaction in bytes.
This prevents DOS attacks because e.g. if we flood the blockchain with $n$ transactions containing tiny balances, we pay $\geq na$ in fees.
In Ethereum, storage costs a fee (denominated in \emph{gas}); if we inflate \inlinesolidity{balances}---which is a finite but unbounded mapping from addresses to integers---with $n$ tiny balances, we may run out of gas.
\end{rmrk}

\subsection{Summary of the critical points}

\begin{enumerate*}
\item
The errors in the Solidity code are not necessarily obvious.
\item
The errors cannot be defended against without rewriting the contract, which requires the seller's cooperation (who controls the contract and might even have planted the bug).
\item
In Plutus, the \emph{buyer} creates a transaction and determines its inputs and outputs.
A transaction might be rejected if the available UTxOs change---but if it succeeds then the buyer knows the outcome.
Thus if the user of a contract anticipates an error (or attack) then they can guard against it independently of the contract's designer.

In Ethereum this is impossible: a buyer can propose a transition to the global virtual machine by calling a function on the Ethereum blockchain, but the designer of the contract designed that function---and because of concurrency, the buyer cannot know the values of the inputs to this function at the time of its execution on the blockchain.
\end{enumerate*}

\section{Conclusions}

We hope this paper will provide a precise yet accessible 
entry point for interested readers, and a useful 
guide to some of the design considerations in this area.

We have seen that Ethereum is an accounts-based blockchain system whereas Cardano (like Bitcoin) is UTxO-based, and we have implemented a specification (Definition~\ref{defn.the.spec}) in Solidity (Section~\ref{sect.solidity.implementation}) and in Plutus (Section~\ref{sect.plutus.implementation}), and we have given a mathematical abstraction of Cardano, \emph{idealised EUTxO } (Section~\ref{sect.idealised.cardano}).
These have raised some surprisingly non-trivial points, both quite mathematical (e.g. Subsection~\ref{subsect.name-carrying})
 and more implementational (e.g. Subsection~\ref{subsect.discussion}), which are discussed in the body of the paper.

The accounts-based paradigm lends itself to an imperative programming style: a smart contract is a \emph{program} that manipulates a global state mapping global variables (`accounts') to values.
The UTxO-based paradigm lends itself to a functional programming style: the smart contract is a \emph{function} that takes a UTxO-list as its input, and returns a UTxO-list as its output, with no other dependencies.
See Theorem~\ref{thrm.might.as.well} and the preceding discussion and lemmas in Subsection~\ref{subsect.observational.equivalence} for a mathematically precise rendering of this intuition, and Subsection~\ref{subsect.name-carrying} for supplementary results.

Smart contracts are naturally concurrent and must expect to scale to thousands, if not billions, of users,\footnote{How to reach distributed consensus in such an environment is another topic, with its own attack surface.  Cardano uses Ouroboros consensus~\cite{KRDO17}.}
but the problems inherent to concurrent imperative programming are well-known and predate smart contracts by decades.
This is an issue with Ethereum/Solidity and it is visible even in our simple example.
Real-life examples, like the infamous DAO hack\footnote{The DAO hack stole approximately 70 million USD from Ethereum, which chose to revert the theft using a hard fork of the Ethereum blockchain~\cite{10.1007/978-3-662-54455-6_8}.},
are clearly related to the problem of transactions having unexpected consequences in a stateful, concurrent environment.

Cardano's UTxO-based structure invites a functional programming paradigm, and this is reflected in Plutus.
The price we pay is arguably an increase in conceptual complexity.
More generally, in Remarks~\ref{rmrk.contracts} and~\ref{rmrk.plutus.error} we discussed how state is handled in Plutus and its UTxO model:
further clarification of the interaction of state with UTxOs is important future work.

In summary: accounts are easier to think about than UTxOs, and imperative programs are easier to read than functional programs, but this ease of use makes us vulnerable to new classes of errors, \emph{especially} in a concurrent safety-critical context.
Such trade-offs may be familiar to many readers with a computer science or concurrency background.


\end{document}